\documentclass[11pt,reqno]{amsart}
\usepackage{amsfonts,amsmath,amssymb}
\newtheorem{thm}{Theorem}[section]
\newtheorem{proposition}[thm]{Proposition}
\newtheorem{corollary}[thm]{Corollary}
\newtheorem{lemma}[thm]{Lemma}

\usepackage{color}
\usepackage{graphicx}
\usepackage{epstopdf}

\title[Dunkl operators at infinity and Calogero-Moser systems]{Dunkl operators at infinity and Calogero-Moser systems }

\author{ A.N. Sergeev}\address{Department of Mathematics, Saratov State University, Astrakhanskaya 83, Saratov, 410012, Russia and Department of Mathematical Sciences,
Loughborough University, Loughborough LE11 3TU, UK}
\email{A.N.Sergeev@lboro.ac.uk}

\author{A.P. Veselov}
\address{Department of Mathematical Sciences,
Loughborough University, Loughborough LE11 3TU, UK  and Moscow State University, Moscow 119899, Russia}
\email{A.P.Veselov@lboro.ac.uk}

\begin{document}

\maketitle

\begin{abstract} 
We define the Dunkl and Dunkl-Heckman operators in infinite number of variables and use them to construct the quantum integrals of the Calogero-Moser-Sutherland problems at infinity. As a corollary we have a simple proof of integrability of the deformed quantum CMS systems related to classical Lie superalgebras. We show how this naturally leads to a quantum version of the Moser matrix, 
which in the deformed case was not known before.
\end{abstract}


\section{Introduction} 

The usual Calogero-Moser, or Calogero-Moser-Sutherland (CMS), system describes the interaction of $N$ particles with equal masses on the line with the inverse square potential or, in trigonometric version, with the inverse $\sin^2$ potential \cite{CMS}. The corresponding quantum Hamiltonian has the form
$$
H_N=-\sum_{i=1}^N\frac{\partial^2}{\partial x^2_i}+\sum_{i<j}\frac{2k(k+1)}{(x_i-x_j)^2}
$$
in the rational case and 
$$
H_N=-\sum_{i=1}^N\frac{\partial^2}{\partial x^2_i}+\sum_{i<j}\frac{2k(k+1)}{\sin^2(x_i-x_j)}
$$ 
in trigonometric case. There is also a very important elliptic case, but we will not consider it in this paper.

The CMS systems admit natural generalizations related to root systems and simple Lie algebras \cite{OP}, and, at the quantum level only, non-symmetric integrable versions called {\it deformed CMS systems} \cite{CFV}, which were shown to be related to basic classical Lie superalgebras in \cite{SV}.
In particular, in the case of Lie superalgebra $\mathfrak{sl} (m,n)$ we have two groups of particles with two different masses described by the following Hamiltonian
$$
H_{n,m} = -\left (\frac{\partial^2}{\partial
x_1^2}+\dots +\frac{\partial^2}{\partial x_{n}^2}\right)-
k\left(\frac{\partial^2}{\partial y_{1}^2}+\dots
+\frac{\partial^2}{\partial y_{m}^2}\right)
+\sum_{i<j}^{n}\frac{2k(k+1)}{\sin^2(x_{i}-x_{j})}
$$ 
\begin{equation}
\label{anmint}
+\sum_{i<j}^{m}\frac{2(k^{-1}+1)}{\sin^2(y_{i}-y_{j})}
+\sum_{i=1}^{n}\sum_{j=1}^{m}\frac{2(k+1)}{\sin^2(x_{i}-y_{j})}.
\end{equation}

The integrability of the deformed CMS systems turned out to be a quite nontrivial question.
 The standard methods like Dunkl operator technique are not working in the general deformed case (for special values of parameters see recent M. Feigin's paper \cite{F}).
For the classical series $A(n,m)$ and $BC(n,m)$ the integrability was proved in \cite{SV} by explicit construction of the quantum integrals. The recurrent procedure was a guesswork based on formulas from Matsuo \cite{Matsuo} and the result was proved by straightforward lengthy calculations. 

The goal of this paper is to give probably simplest explanation of these integrals. Our main tool is the {\it Dunkl operator at infinity}, which seems to be not considered before. We show that although it does not allow to construct the Dunkl operators in the deformed case, it naturally leads to the {\it quantum Moser matrix} for the deformed CMS system. 

This gives an interpretation of the integrals of the deformed CMS systems from \cite{SV} in terms of the {\it quantum Lax pair}, which for the usual CMS system was first considered by Hikami, Ujino and Wadati in \cite{UHW,WHU}. Note that in contrast to the usual case for the deformed CMS systems there is no classical Lax pair to quantize, since their classical counterparts are believed to be non-integrable.

For the deformed CMS system (\ref{anmint}) the quantum Moser matrix is the following $(n+m)\times (n+m)$ matrix 
with the non-commuting entries
\begin{equation}
\label{moserLqint}
L_{ii}=k^{p(i)}\frac{\partial}{\partial x_i},\quad L_{ij}=k^{1-p(j)}\cot(x_i-x_j), \,\, i\neq j
\end{equation}
where $x_{n+j}:=y_j,\,\, j=1,\dots,m$ and  $p(i)=0,\, i=1,\dots,n, \,\, p(i)=1,\,\, i=n+1,\dots, n+m.$
The quantum integrals of (\ref{anmint}) can be constructed as the "deformed total trace" of the powers of $L$
\begin{equation}
\label{qint}
I_r=\sum_{i,j=1}^{n+m} k^{-p(i)}(L^r)_{ij}, \,\,\, r=1,2,\dots.
\end{equation}
(see section 5 below). We present similar formulae in $BC$ case as well.

Another result of the paper is the new formulae for the quantum CMS integrals at infinity in both rational and trigonometric cases for types $A$ and $BC$. In the trigonometric case of type $A$ different formulae for the quantum integrals were recently found in \cite{NS} by Nazarov and Sklyanin. 
Note that only in that case the dependence on the additional parameter $p_0$ can be eliminated (see \cite{SV5}).

\section{Dunkl operator at infinity: rational case}

The usual Dunkl operators in dimension $N$ have the form
\begin{equation}\label{Dun0}
D_{i,N}=\frac{\partial}{\partial x_i}-k\sum_{j\ne i}\frac{1}{x_i-x_j}{(1-\sigma_{ij})},\,\,\, i=1,2,\dots,N,
\end{equation}
where $\sigma_{ij}$ acts on the functions $f(x)$ by permuting variables $x_i$ and $x_j.$
Their main property is the commutativity \cite{Dun}
$$
[D_{i,N},D_{j,N}]=0.
$$

Heckman \cite{H1} made an important observation that the differential operators
  \begin{equation}
  \label{heckcm}
 \mathcal L^{(r)}_{N}=Res \,(D_{1,N}^r+\dots+D^r_{N,N}),
\end{equation}
where $Res$ means the operation of restriction on the space of symmetric polynomials, commute and give 
the integrals for the quantum CMS system. More precisely,   $ \mathcal L^{(2)}_{N}=\mathcal H_N$, 
which is the operator $\mathcal H_N$ is the gauged version of the CMS operator $H_N$ given by 
  \begin{equation}
  \label{CM}
\mathcal H_N=\sum_{i=1}^n\frac{\partial^2}{\partial x^2_i}-\sum_{i<j}\frac{2k}{x_i-x_j}\left(\frac{\partial}{\partial x_i}-\frac{\partial}{\partial x_j}\right).
\end{equation}

The operator (\ref{CM}) preserves the algebra of symmetric polynomials 
$$\Lambda_{N}=\Bbb C[x_{1},\dots, x_{N}]^{S_N},$$
generated (not freely) by $p_j(x) = x_1^j + \dots + x_N^j, \, j \in \mathbb Z_{\ge 0}$ 

Let $\Lambda$ be the {\it algebra of symmetric functions} defined as the inverse limit of $\Lambda_N$ in the category of graded algebras (see \cite{Ma}).
We will consider the larger algebra $\bar\Lambda=\Lambda[p_0],$ which is the commutative algebra with the free generators $p_{i},\, i\in \mathbb Z_{\ge 0}.$ The dimension $p_0=1+1+\dots +1=N$ does not make sense in infinite-dimensional case, so we have added $p_0$ as {\it an additional variable}.
$\bar\Lambda$ has a natural grading, where the degree of $p_i$ is $i.$ 

For every natural $N$ there is a homomorphism $\varphi_N:  \bar\Lambda \rightarrow \Lambda_{N}:$
\begin{equation}\label{phin}
\varphi_N (p_j)= x_1^j + \dots + x_N^j, \, j \in \mathbb Z_{\ge 0}.
\end{equation}

Define the {\it infinite dimensional Dunkl operator} $D_{k}: \bar\Lambda[x]\rightarrow \bar\Lambda[x]$  by 
\begin{equation}\label{Duninf}
D_{k}=\partial-k\Delta,
\end{equation}
where the differentiation $\partial$ in  $\bar\Lambda[x]$ is defined by the formulae
$$
\partial(x)=1, \,\, \partial (p_l)=lx^{l-1},  l\in \Bbb Z_{\ge 0},
$$
and the operator $\Delta: \bar\Lambda[x] \rightarrow \bar\Lambda[x]$ is defined by
$$
\Delta(x^lf)=\Delta(x^l)f,\,\,\,  \Delta(1)=0, \,\,\,f\in \bar\Lambda ,\,l\in \Bbb Z_{\ge 0}
$$
and
$$
\Delta(x^l)=x^{l-1}p_0+x^{l-2}p_1+\dots+ xp_{l-2}+p_{l-1}-lx^{l-1},\,\,  l>0.
$$

The motivation is given by the following proposition. 
 Let $\varphi_{i,N} : \bar\Lambda[x]\longrightarrow \Lambda_{N}[x_i]$ be the homomorphism such that
$$
\varphi_{i,N}(x)=x_i,\,\varphi_{i,N}(p_l)=x_1^l+\dots+x_N^l, \, l \in \mathbb Z_{\ge 0}.
$$

\begin{proposition}
\label{dcomm}
The following diagram
\begin{equation}\label{commutdun}
\begin{array}{ccc}
\bar\Lambda[x]&\stackrel{D_{k}}{\longrightarrow}&\bar\Lambda[x]\\ \downarrow
\lefteqn{\varphi_{i,N}}& &\downarrow \lefteqn{\varphi_{i,N}}\\
\Lambda_{N}[x_i]&\stackrel{D_{i,N}}{\longrightarrow}& 
\Lambda_{N}[x_i],\\
\end{array}
\end{equation}
 where $D_{i,N}$ are the Dunkl operators (\ref{Dun0}), is commutative.
 \end{proposition}

\begin{proof}
We have
$$
\varphi_{i,N} \circ \Delta=\left(\sum_{j\ne i}\frac{1}{x_i-x_j}(1-\sigma_{ij})\right) \circ \varphi_{i,N} 
$$
since 
$$\sum_{j\ne i}\frac{1}{x_i-x_j}(1-\sigma_{ij})x_i^l=\sum_{j\ne i}\frac{x_i^l-x^l_j}{x_i-x_j}
=x_i^{l-1}N+x_i^{l-2}p_1+\dots+ x_ip_{l-2}+p_{l-1}-lx_i^{l-1}.$$
The realtion 
$
\varphi_{i,N} \circ \partial=\partial_i \circ \varphi_{i,N} 
$
is obvious.
\end{proof}

 Introduce also a linear operator $E : \bar\Lambda[x]\longrightarrow \bar\Lambda$  by the formula 
 \begin{equation}\label{E}
  E(x^lf)=p_lf,\, f\in\bar\Lambda,\,\,l\in \Bbb Z_{\ge 0}
 \end{equation}
and define the operators $\mathcal {L}^{(r)}_{k}: \bar\Lambda\longrightarrow \bar\Lambda,\,\, r \in \mathbb Z_+$  by
 \begin{equation}\label{Lr}
 \mathcal{L}^{(r)}=E\circ D_{k}^r,
 \end{equation}
 where the action of the right hand side is restricted to $\bar\Lambda.$

We claim that these operators give the quantum CMS integrals at infinity.
More precisely, we have the following result.

\begin{thm}\label{Heck} 
The differential operators $\mathcal{L}^{(r)}$ commute with each other:
$$[\mathcal{L}^{(r)}, \mathcal{L}^{(s)}]=0.$$
The operator $\mathcal L^{(2)}$ has the following explicit form:
 \begin{equation}\label{inf1ratio}
\mathcal L^{(2)}=\sum_{a,b\geq 1}p_{a+b-2}\partial_{a}\partial_{b}-k\sum_{a,b\geq 0}p_{a}p_b \partial_{a+b+2} +(1+k)\sum_{a \geq 2}(a-1) p_{a-2}\partial_a
\end{equation}
with  $\partial_a = a\partial/\partial p_a,$ and coincides with the rational CMS operator at infinity.
\end{thm}

\begin{proof} Consider $f\in \Lambda.$ Since $E$ and $\Delta$ commute with multiplication by $f$, we have 
$$
ad(f)^{r+1}(E\circ D_{k}^r)=E\circ ad(f)^{r+1}(D_{k}^r)
$$
and therefore
$$
 ad(f)^{r+1}(D_{k,p_0}^r)= ad(f)^{r+1}(\partial^r)=0,
$$
which shows that $\mathcal{L}^{(r)}$ is a differential  operator of order $r$.
The explicit form (\ref{inf1ratio}) easily follows from a direct calculation.

In order to prove the commutativity we consider the finite dimensional reductions. 
We have the  following commutative  diagram 
 \begin{equation}\label{commdiaheck}
\begin{array}{ccc}
\bar\Lambda&\stackrel{E\circ D_{k}^r}{\longrightarrow}&\bar\Lambda\\ \downarrow
\lefteqn{\varphi_{N}}& &\downarrow \lefteqn{\varphi_{N}}\\
\Lambda_N&\stackrel{\mathcal L^{(r)}_{N}}{\longrightarrow}& 
\Lambda_N, \\
\end{array}
\end{equation}
where $\mathcal L^{(r)}_{N}$ are the CMS integrals given by Heckman's construction (\ref{heckcm}) and the homomorphism $\varphi_{N}:\Lambda \rightarrow \Lambda_N$ is defined by 
\begin{equation}\label{varphin}
\varphi_N(p_l)=x_1^l+\dots+x_N^l, \, l \in \mathbb Z_{\ge 0}.
 \end{equation}
 Indeed, for any $f\in \bar\Lambda$ we have 
 $
 D_{k}^r(f)=\sum_{l}x^lg_l,\,\,\, g_l\in \bar\Lambda,
 $
 where the sum is finite. We have by proposition \ref{dcomm}
 $$
 D_{i,N}^r\circ\varphi_N(f)=\varphi_{i,N}\circ D_{k}^r(f)=\sum_{l}x_i^l\varphi_N(g_l),
 $$
 $$
 \sum_{i=1}^ND_{i,N}^r\circ\varphi_N(f)=\sum_{i=1}^N\sum_{l}x_i^l\varphi_N(g_l)=\sum_{l}\varphi_{N}(p_l)\varphi_N(g_l)=\varphi_N(E(D_{k}^r(f))),
 $$
 which proves the commutativity of the diagram.
This implies that 
$$
\varphi_N([\mathcal L^{(r)},\mathcal L^{(s)}](f))=[\mathcal L_{N}^{(r)},\mathcal L_{N}^{(s)}](\varphi_N(f))=0
$$
since the integrals (\ref{heckcm}) commute \cite{H}. To conclude the proof we need the following

\begin{lemma} 
\label{edin} Let $f$ be an element of $\bar\Lambda.$ If $\varphi_N(f)=0$ for all $N$ then $f =0.$
 \end{lemma}
\begin{proof} 
By definition $f$ is a polynomial in a finite number $M$ of generators $p_r,\, r\in\Bbb Z_{> 0}$ with the coefficients, polynomially depending on $p_0.$ Take $N$ bigger than this number $M$. Since the corresponding $\varphi_N(p_{r})$ are algebraically independent and $\varphi_N(f)=0$, all the coefficients of $f$ are zero at $p_0=N$. Since this is true for all $N>M$ the coefficients are identically zero, and therefore $f=0.$
\end{proof}

Applying lemma we have the commutativity $
[\mathcal L^{(r)},\mathcal L^{(s)}]=0.$
 \end{proof}
 
\section{Deformed CMS operators: rational case}

The deformed CMS operators in the rational case have the form

$$
H_{n,m} = \sum_{i=1}^n\frac{\partial^2}{\partial x^2_i}
+k\sum_{i=1}^m\frac{\partial^2}{\partial y^2_i}$$
\begin{equation}
\label{anm}
-\sum_{i<j}^{n}\frac{2k(k+1)}{(x_{i}-x_{j})^2} -\sum_{i<j}^{m}\frac{2(k^{-1}+1)}{(y_{i}-y_{j})^2}
-\sum_{i=1}^{n}\sum_{j=1}^{m}\frac{2(k+1)}{(x_{i}-y_{j})^2}.
\end{equation}
They describe the interaction of two groups of particles on the line with masses 1 and $1/k$ respectively.
When $k=1$ we have the usual CMS system with $n+m$ particles, hence the terminology.

The operators (\ref{anm}) for $m=1$ were introduced and studied by Chalykh, Feigin and one of the authors in \cite{CFV},
for general $m$ they were considered by Berest and Yakimov \cite{Berest}. Their integrability was first proved in \cite{SV}, 
where the quantum integrals were constructed by a recursive procedure.
We are going to show now that this procedure has a simple explanation in terms of Dunkl operators.

Let  ${\mathcal H}_{n,m}=\Psi_0 L_{n,m}\Psi_0^{-1}$ be the gauged form of (\ref{anm}) with
$$\Psi_0=\prod_{i<j}^n(x_i-x_j)^k\prod_{i<j}^m(y_i-y_j)^{\frac{1}{k}}\prod_{i}^n\prod_{j}^n(x_i-y_j):$$
$$
 {\mathcal H}_{n,m}=\sum_{i=1}^n\frac{\partial^2}{\partial x^2_i}
+k\sum_{i=1}^m\frac{\partial^2}{\partial y^2_i}-\sum_{i < j}^n
\frac{2k}{x_{i}-x_{j}}\left(\frac{\partial}{\partial x_{i}}-\frac{\partial}{\partial x_{j}}\right) 
$$
\begin{equation}
\label{defCMS}
 -\sum_{i <
j}^m \frac{2}{y_{i}-y_{j}}\left(
\frac{\partial}{\partial y_{i}}-
\frac{\partial}{\partial y_{j}}\right)-\sum_{i=1 }^n\sum_{
j=1}^ m \frac{2}{x_{i}-y_{j}}\left(
\frac{\partial}{\partial x_{i}}-k
\frac{\partial}{\partial y_{j}}\right).
\end{equation}

It is convenient to denote $y_j=x_{n+j},\,\, j=1,\dots,m$ and introduce the parity function $p(i)=0,\, i=1,\dots,n$ and
$p(i)=1,\,\, i=n+1,\dots, n+m.$

Following \cite{SV} consider the operators  $\partial^{(r)}_i$ defined recursively by
\begin{equation}
\label{recurdef}
 \partial^{(r)}_i=\partial^{(1)}_i\partial^{(r-1)}_i-\sum_{j\ne i}\frac{k^{1-p(j)}}{x_i-x_j}(\partial^{(r-1)}_i-\partial^{(r-1)}_j)
\end{equation}
 with $\partial^{(1)}_i=k^{p(i)}\frac{\partial}{\partial x_i}$ (cf. formula (17) in \cite{SV}).
 One can easily check that the operator
 $$\mathcal L_{n,m}^{(2)}=\sum_{i=1}^{n+m}k^{p(i)} \partial^{(2)}_i$$ coincides with deformed CMS operator (\ref{defCMS}).
 In \cite{SV} it was proved by a lengthy but direct calculation that the operators
 \begin{equation}
\label{defin}
\mathcal L_{n,m}^{(r)}=\sum_{i=1}^{n+m}k^{p(i)} \partial^{(r)}_i
\end{equation}
commute with each other, and in particular are the quantum integrals of the deformed CMS system.
The recursion formulae (\ref{recurdef}) were guess-work based on the recursive Matsuo's formulae \cite{Matsuo}.

Now we are going to give a much simpler proof of this together with more conceptual explanation of these formulae.

Define  $\varphi^{(i)}_{n,m}:\bar\Lambda[x]\longrightarrow \Bbb C[x_1,\dots,x_{n+m}]$ by 
$\varphi^{(i)}_{n,m}(x)=x_i$ and 
\begin{equation}
\label{varp}
\varphi^{(i)}_{n,m}(p_l)= p_l(x,k):=\sum_{i=1}^{n+m} k^{-p(i)}x_i^l=\sum_{i=1}^{n} x_i^l + \frac{1}{k}\sum_{i=n+1}^{n+m} x_i^l,
\end{equation}
for all $l\in \mathbb Z_{\ge 0}.$

Denote by $\Lambda_{n,m}$ the subalgebra in $\Bbb C[x_1,\dots,x_{n+m}]$ generated by the deformed power sums $p_l(x,k), \, l\in \mathbb Z_{>0}.$
We will show that the operators  $\partial^{(r)}_i$ maps the algebra $\Lambda_{n,m}$  into $\Lambda_{n,m}[x_i]$ (see diagram (\ref{commdia}) below).

In the deformed case we do not have the commutative diagram similar to (\ref{commutdun}), 
but we have the following important relation.
 
 \begin{proposition} The following relation holds on  $\Lambda[x]$:
 \begin{equation}
\label{commnet}
 \varphi_{n,m}^{(i)}\circ D=k^{p(i)}\frac{\partial}{\partial x_i}\circ \varphi_{n,m}^{(i)}-\sum_{j\ne i}\frac{k^{1-p(j)}}{x_i-x_j}(\varphi_{n,m}^{(i)}-\varphi^{(j)}_{n,m}).
 \end{equation}
 \end{proposition}
\begin{proof}  For any $f \in \Lambda$ we have
$$
\varphi^{(i)}_{n,m}\circ(\partial-k\Delta)(x^lf)=\varphi_{n,m}^{(i)}(lx^{l-1}f+x^l\partial f
$$
$$
-k(lx^{l-1}p_0+x^{l-2}p_1+\dots+xp_{l-2}+p_{l-1}-lx^{l-1})f)
$$
$$
=lx_i^{l-1}(1+k)\varphi_{n,m}(f)+x_i^l\varphi_{n,m}(\partial f)
$$
$$
-k(x_i^{l-1}p_0+x_i^{l-2}p_1+\dots+x_ip_{l-2}+p_{l-1})\varphi^{(i)}_{n,m}(f).
$$
On the other hand
$$
k^{p(i)}\frac{\partial}{\partial x_i}\circ \varphi_{n,m}^{(i)}(x^lf)-\sum_{j\ne i}\frac{k^{1-p(j)}}{x_i-x_j}(\varphi_{n,m}^{(i)}-\varphi^{(j)}_{n,m})(x^lf)
$$
$$
=k^{p(i)}lx_i^{l-1}\varphi^{(i)}_{n,m}(f)+x_i^lk^{p(i)}\partial_i(\varphi^{(i)}_{n,m}(f))-
(x_i^{l-1}(kn+m)
$$
$$
-k(x_i^{l-2}p_1+\dots+xp_{l-2}+p_{l-1}-k^{-p(i)}lx^{l-1})\varphi^{(i)}_{n,m}(f)
$$
$$
=(k^{(p(i)}+k^{1-p(i)})lx_i^{l-1}\varphi^{(i)}_{n,m}(f)+k^{p(i)}x_i^l\partial_i(\varphi^{(i)}_{n,m}(f))
$$
$$
-k(x_i^l(n+k^{-1}m)+x_i^{l-2}p_1+\dots+xp_{l-2}+p_{l-1}))\varphi^{(i)}_{n,m}(f).
$$
Since $k^{p(i)}+k^{1-p(i)}=1+k$ for all $i=1,\dots,n+m$ we only need to show that
\begin{equation}
\label{varp}
\varphi^{(i)}_{n,m}(\partial f)=k^{p(i)}\partial_i\varphi^{(i)}_{n,m}(f).
\end{equation}
Since both $\partial$ and $\partial_i$ are differentiations it is enough to check this for  $f=p_l$, which is obvious. 
\end{proof}

\begin{proposition} 
\label{propcomm}
The following diagram is commutative 
\begin{equation}
\label{commdia}
\begin{array}{ccc}
\bar\Lambda&\stackrel{D^r}{\longrightarrow}&\bar\Lambda[x]\\ \downarrow
\lefteqn{\varphi^{(i)}_{n,m}}& &\downarrow \lefteqn{\varphi^{(i)}_{n,m}}\\
\Lambda_{n,m}&\stackrel{\partial^{(r)}_i}{\longrightarrow}& 
\Lambda_{n,m}[x_i].
\end{array}
 \end{equation}
\end{proposition} 
\begin{proof}  Induction in $r.$ When $r=1$ this follows from (\ref{varp}).
For $r>1$ we have
$$
 \varphi_{n,m}^{(i)}(D^r(f))= \varphi_{n,m}(DD^{r-1}(f))= \varphi_{n,m}(D(g))
$$
where $g=D^{r-1}f\in \bar\Lambda[x]$. By previous proposition
$$
\varphi_{n,m}^{(i)}(D(g))=\partial_i^{(1)}\circ \varphi_{n,m}^{(i)}(g)-\sum_{j\ne i}\frac{k^{1-p(j)}}{x_i-x_j}(\varphi_{n,m}^{(i)}(g)-\varphi^{(j)}_{n,m}(g)).
$$
By inductive assumption $\varphi_{n,m}^{(i)}(g)=\partial^{(r-1)}_i \varphi_{n,m}^{(i)}(f)$ and thus
$$
\varphi_{n,m}^{(i)}(D(g))= \left(\partial^{(1)}_i\partial^{(r-1)}_i-\sum_{j\ne i}\frac{k^{1-p(j)}}{x_i-x_j}(\partial^{(r-1)}_i-\partial^{(r-1)}_j)\right)\varphi^{(i)}_{n,m}(f),
$$
where we have used that for $f \in \bar\Lambda$ $\varphi^{(i)}_{n,m}(f)=\varphi^{(j)}_{n,m}(f).$ 
Thus 
$$
\varphi_{n,m}^{(i)}(D^r(f))= \partial^{(r)}_i\varphi^{(i)}_{n,m}(f),
$$
which concludes the proof.
\end{proof} 

Define the homomorphism $\varphi_{n,m}: \bar \Lambda \rightarrow \Lambda_{n,m}$ by 
$$\varphi_{n,m}(p_l)=\sum_{i=1}^{n+m} k^{-p(i)}x_i^l, \,\,\, l\in \mathbb Z_{\ge 0}.$$

\begin{thm}
\label{intrat}
The following diagram is commutative
 \begin{equation}
\label{commdiadef}
\begin{array}{ccc}
\bar\Lambda&\stackrel{\mathcal L^{(r)}}{\longrightarrow}&\bar\Lambda\\ \downarrow
\lefteqn{\varphi_{n,m}}& &\downarrow \lefteqn{\varphi_{n,m}}\\
\Lambda_{n,m}&\stackrel{\mathcal L_{n,m}^{(r)}}{\longrightarrow}& 
\Lambda_{n,m},
\end{array}
\end{equation}
where $\mathcal L_{n,m}^{(r)}$ are the operators (\ref{defin}). In particular, 
these operators commute:
$$
[\mathcal L_{n,m}^{(r)}, \mathcal L_{n,m}^{(s)}]=0
$$
for all $r,s \in \mathbb Z_{>0},$ and thus give the quantum integrals of the deformed CMS system.
 \end{thm} 
 \begin{proof} For any $f\in \Lambda$ we have
 $$
 D^r(f)=\sum_{l}x^lg_l,\,\,\, g_l\in\Lambda,
 $$
 where the sum in the right hand side is finite. By proposition \ref{propcomm}
  $$
 \partial_i^{(r)}(\varphi_{n,m}(f))=\varphi^{(i)}_{n,m}( D^r(f))=\sum_{l}x_i^l\varphi_{n,m}(g_l).
 $$
 Hence we have
 $$
 \mathcal L_{n,m}^{(r)}\varphi_{n,m}(f)=\sum_{i=1}^{n+m}k^{p(i)} \partial^{(r)}_i\varphi_{n,m}(f)= \sum_{i=1}^{n+m}k^{p(i)} \sum_{l}x_i^l\varphi_{n,m}(g_l)
 $$
 $$
 =\sum_{l}\varphi_{n,m}(p_l(x,k))\varphi_{n,m}(g_l)=\varphi_{n,m} (\mathcal L^{(r)} (f)).
 $$
 This proves the commutativity of the diagram. The commutativity of the operators (\ref{defin}) now follows from 
 the commutativity of the CMS operators $\mathcal L^{(r)}$ at infinity.
 \end{proof}

\section{Quantum Moser matrix for the deformed CMS system} 

In contrast to the usual CMS case the classical version of the 
deformed CMS system is not integrable, see \cite{CFV}.
This means that there is no proper replacement for Moser's matrix 
$$
L= \begin{pmatrix}
p_1  &  \frac{k}{q_1-q_2} &   \frac{k}{q_1-q_3}   &  \dots  & \dots &  \frac{k}{q_1-q_n}   &  \\
-\frac{k}{q_1-q_2}  &  p_2 &  \frac{k}{q_2-q_3}  & & \dots  & \frac{k}{q_2-q_n}   & \\
\dots   & \dots & \dots &\dots &  \dots &  \dots &\\
-\frac{k}{q_1-q_n}    & -\frac{k}{q_2-q_n} & -\frac{k}{q_3-q_n}& \dots & -\frac{k}{q_{n-1}-q_n}  & p_n
\end{pmatrix}
$$
Recall that Moser has shown that the equations of motion of the classical CMS system with
$$H=\sum_{i=1}^n p_i^2 -\sum_{i<j}^{n} \frac {2k^2}{(q_i-q_j)^2}$$
can be rewritten in the Lax form as
$$\dot L = [L,M],$$
where 
$$
M= -2k\begin{pmatrix}
a_{11}  &  \frac{1}{(q_1-q_2)^2} &  \frac{1}{(q_1-q_3)^2}   &  \dots  & \dots  &\frac{1}{(q_1-q_n)^2}  &  \\
\frac{1}{(q_1-q_2)^2}  &  a_{22} &   \frac{1}{(q_2-q_3)^2}   & \dots &  \dots & \frac{1}{(q_2-q_n)^2}& \\
 \dots  & \dots & \dots &\dots &  \dots &   \dots & \\
\frac{1}{(q_1-q_n)^2}    & \frac{1}{(q_2-q_n)^2} & \frac{1}{(q_3-q_n)2}& \dots & \frac{1}{(q_{n-1}-q_n)^2}  & a_{nn}
\end{pmatrix}
$$
with $a_{ii}=-\sum_{i\neq j}^n  \frac{1}{(q_i-q_j)^2}.$ Note that the last condition means that
$Me=0$, where $e=(1,\dots,1)^T.$

The Dunkl operator approach from the previous section naturally leads to the following
{\it quantum Moser matrix} $L$ for the deformed CMS system. The quantum analogue of the Lax pair for the usual CMS systems was first proposed by Wadati, Hikami and Ujino \cite{UHW,WHU}. 

Consider the following $(n+m)\times (n+m)$ matrix with the (non-commuting) entries
\begin{equation}
\label{moserLq}
L_{ii}=k^{p(i)}\frac{\partial}{\partial x_i},\quad L_{ij}=\frac{k^{1-p(j)}}{x_i-x_j}, \,\, i\neq j.
\end{equation}
motivated by formula (\ref{commnet}). Note that the operators $$k^{p(i)}\partial_i-\sum_{j\ne i}\frac{k^{1-p(j)}}{x_i-x_j}$$ and $k^{p(i)}\partial_i$ are gauge equivalent with $$\Psi_0=\prod_{i<j}^{n+m}(x_i-x_j)^{k^{1-p(i)-p(j)}}.$$

Introduce also $(n+m)\times (n+m)$ matrix $M$ by
\begin{equation}
\label{moserMq}
M_{ij}=\frac{2k^{1-p(j)}}{(x_i-x_j)^2}, \,\, i\neq j, \quad M_{ii}=-\sum_{j\neq i}^{n+m}\frac{2k^{1-p(j)}}{(x_i-x_j)^2}.
\end{equation}
Note that this matrix has the properties $Me=0$ (like in the usual case), 
and $e^*M=0$, where $e^*=(1,\dots, 1, \frac{1}{k}, \dots,\frac{1}{k})$, or, more precisely,
$e^*_i=k^{-p(i)}.$

Introduce also the "matrix Hamiltonian" 
$H$ which is diagonal $(n+m)\times (n+m)$ matrix with the deformed CMS operator (\ref{anm}) on the diagonal:
$$H_{ii}=H_{n,m},\,\, H_{ij}=0,\,\,  i\neq j$$
The commutator $[L, H]$ has the entries $[L, H]_{ij}=[L_{ij}, H_{n,m}],$ which can be considered as a quantum version of $\dot L$ (cf. \cite{UHW,WHU}).

\begin{thm} (Quantum Lax pair for the deformed CMS system)
We have the following identity
\begin{equation}
\label{quantumLpair}
[L, H]=[L, M].
\end{equation}
\end{thm} 
 The proof can be done by direct calculation.

\begin{corollary}
\label{corol}
The operators $L_{n,m}^{(r)}=e^*L^r e$ are the quantum integrals of the deformed CMS system (\ref{anm}). 
These integrals coincide with the integrals from the previous section modulo gauge transformation.
\end{corollary} 

Indeed, from (\ref{quantumLpair}) we have $[L, H-M]=0$, and hence $[L^r, H-M]=0$, or
$$[L^r, H]=[L^r, M].$$
This implies
$$[L_{n,m}^{(r)}, H_{n,m}]=0,$$
since $Me=e^* M=0$ (cf. the case of usual CMS system in \cite{WHU}).
This proves that $L_{n,m}^{(r)}$ are the integrals of the deformed CMS system. One can check that
$L_{n,m}^{(2)}=H_{n,m}.$ Note that the integrals $L_{n,m}^{(r)}$ can be interpreted as the "deformed total trace" of the powers of quantum Moser's matrix:
$$L_{n,m}^{(r)}=\sum_{i,j=1}^{m+n}k^{-p(i)}(L^R)_{ij}.$$

The fact that $L_{n,m}^{(r)}$ commute with each other does not follow from the Lax approach.
In our case this follows from the results of the previous section since $L_{n,m}^{(r)}$ are the gauged versions of $\mathcal L_{n,m}^{(r)}$.

\section{Trigonometric case: Dunkl-Heckman operator at infinity} 

In this section we follow mainly to our paper \cite{SV6}, where a more general Laurent case is considered.
Since it is largely parallel to the rational case we will omit most of the proofs.
We will be also using the same letters to denote the similar quantities as in the rational case; 
hopefully this will not lead to much of confusion.

In the trigonometric (hyperbolic) case we have the following CMS operator
$$H_N = \sum_{i=1}^N
\frac{\partial^2}{\partial
z_{i}^2}-\sum_{i<j}^N \frac{2k(k+1)}{\sinh
^2(z_{i}-z_{j})}.$$ 
 It has an eigenfunction $$\Psi_0= \prod_{i<j}^N \sinh^{-k} (z_i-z_j)$$ with the eigenvalue 
$\lambda_0= -k^2N(N-1)/4.$
Its gauged version $\frac{1}{4}\Psi_0^{-1} (L_N-\lambda_0) \Psi_0$ in the exponential coordinates 
$x_i = e^{2z_i}$ has the form
\begin{equation}
\label{CMtrig}
 {\mathcal H}_{N}=\sum_{i=1}^N
\left(x_{i}\frac{\partial}{\partial
x_{i}}\right)^2-k\sum_{ i < j}^N
\frac{x_{i}+x_{j}}{x_{i}-x_{j}}\left(
x_{i}\frac{\partial}{\partial x_{i}}-
x_{j}\frac{\partial}{\partial
x_{j}}\right).
\end{equation}


 The corresponding version of the Dunkl operator in this case was first introduced by Heckman \cite{H} and has the form
 \begin{equation}
  \label{heckdun}
 D_{i,N}=\partial_i-\frac{k}{2}\sum_{j\ne i}^N\frac{x_i+x_j}{x_i-x_j}(1-\sigma_{ij}), \quad \partial_i=x_i\frac{\partial}{\partial x_i}, \, \, i=1,\dots, N,
 \end{equation}
 where as before $\sigma_{ij}$ is a transposition, acting on the functions by permuting the coordinates $x_i$ and $x_j.$
 The main problem with these operators is that they do not commute. However, Heckman
 \cite{H} managed to show that the differential operators 
  \begin{equation}
  \label{heckcmtrig}
 \mathcal L^{(r)}_{N}=Res \,(D_{1,N}^r+\dots+D^r_{N,N}),
\end{equation}
where $Res$ means the operation of restriction on the space of symmetric polynomials, do commute with each other
\begin{equation}
  \label{heckcom}
[\mathcal L^{(r)}_{N}, \mathcal L^{(s)}_{N}]=0.
\end{equation}
Since $\mathcal L^{(2)}_{N}= {\mathcal H}_{N}$ they are 
the integrals of the quantum CMS system (\ref{CMtrig}).
 
The operator
 \begin{equation}\label{heckdel}
\Delta_{i,N}:=\sum_{j\ne i}^N\frac{x_i+x_j}{x_i-x_j}(1-\sigma_{ij})
\end{equation}
acts trivially on the algebra of symmetric polynomials $\Lambda_{N}$ and has the property
$$
\Delta_{i,N}(x_i^l)=\sum_{j\ne i}\frac{x_i+x_j}{x_i-x_j}(1-\sigma_{ij})(x_i^l)=\sum_{j\ne i}\frac{x_i+x_j}{x_i-x_j}(x_i^l-x^l_j)
$$
\begin{equation}
 \label{identity}
=x_i^lN+2x_i^{l-1}p_1+\dots+ 2x_ip_{l-1}+p_l-2lx_i^l.
\end{equation}

Define the {\it infinite dimensional Dunkl-Heckman operator} $D_{k}: \bar\Lambda[x]\rightarrow \bar\Lambda[x]$  by 
\begin{equation}\label{DuninfH}
D_{k}=\partial-\frac12k\Delta,
\end{equation}
where the differentiation $\partial$ in  $\bar\Lambda[x]$ is defined by the formulae
$$
\partial(x)=x, \,\, \partial (p_l)=lx^{l},  l\in \Bbb Z_{\ge 0},
$$
and the operator $\Delta: \bar\Lambda[x] \rightarrow \bar\Lambda[x]$ is defined by
$$
\Delta(x^lf)=\Delta(x^l)f,\,\,\,  \Delta(1)=0, \,\,\,f\in \bar\Lambda ,\,l\in \Bbb Z_{\ge 0}
$$
and
$$
\Delta(x^l)=x^lp_0+2x^{l-1}p_1+\dots+ 2xp_{l-1}+p_l-2lx^l,\,\,  l>0,\,\,  l>0.
$$

One can check that the following diagram
$$
\begin{array}{ccc}
\bar\Lambda[x]&\stackrel{D_{k}}{\longrightarrow}&\bar\Lambda[x]\\ \downarrow
\lefteqn{\varphi_{i,N}}& &\downarrow \lefteqn{\varphi_{i,N}}\\
\Lambda_{N}[x_i]&\stackrel{D_{i,N}}{\longrightarrow}& 
\Lambda_{N}[x_i],\\
\end{array}
$$
 is commutative, where $D_{i,N}$ are the Dunkl-Heckman operators (\ref{heckdun}), and
 $\varphi_{i,N}(x)=x_i$, $\varphi_{i,N}(p_l)=x_1^l+\dots+x_N^l,\, \, l \ge 0,$ as before.

Let $E : \bar\Lambda[x]\longrightarrow \bar\Lambda$ be the same as above: 
$  E(x^lf)=p_lf,\, f\in\bar\Lambda,\,\,l\in \Bbb Z_{\ge 0}.$
Define the operators $\mathcal {L}^{(r)}: \bar\Lambda\longrightarrow \bar\Lambda,\,\, r \in \mathbb Z_+$  by
 \begin{equation}\label{LrH}
 \mathcal{L}^{(r)}=E\circ D_{k}^r,
 \end{equation}
 where the action of the right hand side is restricted to $\bar\Lambda.$
 
 The operator $\mathcal L^{(2)}$ has the following explicit form
  \begin{equation}\label{inf1}
\mathcal L^{(2)}=\sum_{a,b>0}p_{a+b}\partial_{a}\partial_{b}-k\sum_{a,b>0}p_{a}p_b \partial_{a+b}+(1+k)\sum_{a>0}a p_a\partial_a- kp_0 \sum_{a>0} p_{a} \partial_{a},
\end{equation}
where $\partial_a = a\frac{\partial}{\partial p_a},$ and is known to be the (trigonometric) CMS operator at infinity (see \cite{Stanley, Awata, SV5}).

Note that the dependence on $p_0$ in the trigonometric case can be easily eliminated since $\sum_{a>0}p_a\partial_a$ is the total momentum, which corresponds to the stability property of the CMS operator in this case (see the discussion in \cite{SV5}).

The claim is that the operators (\ref{LrH}) commute:
 \begin{equation}\label{LrHcom}
 [\mathcal{L}^{(r)},\mathcal{L}^{(s)}]=0,
 \end{equation}
and thus are the quantum CMS integrals at infinity. This follows from the commutativity of Heckman's integrals (\ref{heckcm}), lemma \ref{edin} and the commutativity of the diagram
 $$
\begin{array}{ccc}
\bar\Lambda&\stackrel{E\circ D_{k}^r}{\longrightarrow}&\bar\Lambda\\ \downarrow
\lefteqn{\varphi_{N}}& &\downarrow \lefteqn{\varphi_{N}}\\
\Lambda_N&\stackrel{\mathcal L^{(r)}_{N}}{\longrightarrow}& 
\Lambda_N, \\
\end{array}
$$
where $\mathcal L^{(r)}_{N}$ are the CMS integrals given by (\ref{heckcm}) and the homomorphism $\varphi_{N}:\bar\Lambda \rightarrow \Lambda_N$ is defined by 
$\varphi_N(p_l)=x_1^l+\dots+x_N^l, \, l \ge 0.$

Consider now the deformed CMS operator \cite{SV},
which in the exponential coordinates has the form:
$$
H_{n,m}=\sum_{i=1}^n
\left(x_{i}\frac{\partial}{\partial
x_{i}}\right)^2+k\sum_{j=1}^m
\left(y_{j}\frac{\partial}{\partial
y_{j}}\right)^2
-\sum_{i<j}^{n}\frac{2k(k+1)x_ix_j}{(x_{i}-x_{j})^2}
$$
\begin{equation}
\label{anmi} 
-\sum_{i<j}^{m}\frac{2(k^{-1}+1)y_iy_j}{(y_{i}-y_{j})^2}
-\sum_{i=1}^{n}\sum_{j=1}^{m}\frac{2(k+1)x_i y_j}{(x_{i}-y_{j})^2},
\end{equation}
or, in the gauged form by 
\begin{equation}
\label{defCMSt}
 {\mathcal H}_{n,m}=\sum_{i=1}^n
\left(x_{i}\frac{\partial}{\partial
x_{i}}\right)^2+k\sum_{j=1}^m
\left(y_{j}\frac{\partial}{\partial
y_{j}}\right)^2-k\sum_{1\le i < j\le n}
\frac{x_{i}+x_{j}}{x_{i}-x_{j}}\left(
x_{i}\frac{\partial}{\partial x_{i}}-
x_{j}\frac{\partial}{\partial x_{j}}\right)$$ $$- \sum_{1\le i <
j\le m} \frac{y_{i}+y_{j}}{y_{i}-y_{j}}\left(
y_{i}\frac{\partial}{\partial y_{i}}-
y_{j}\frac{\partial}{\partial y_{j}}\right)-\sum_{i=1 }^n\sum_{
j=1}^ m \frac{x_{i}+y_{j}}{x_{i}-y_{j}}\left(
x_{i}\frac{\partial}{\partial x_{i}}-k
y_{j}\frac{\partial}{\partial y_{j}}\right).
\end{equation}

We use the notations from Section 3: $y_j=x_{n+j},\,\, j=1,\dots,m$ and  $p(i)$ be the parity function,
$\varphi^{(i)}_{n,m} :\Lambda[x]\longrightarrow \Bbb C[x_1,\dots,x_{n+m}]$ be defined by (\ref{varp}) 
and $\Lambda_{n,m}$ be the subalgebra in $\Bbb C[x_1,\dots,x_{n+m}]$ generated by 
the deformed power sums $$p_l(x,k)=\sum_{i=1}^{n} x_i^l + \frac{1}{k}\sum_{i=n+1}^{n+m} x_i^l.$$ 

One can check that  the following equality is valid on $\Lambda[x]$:
 \begin{equation}\label{qMtl}
 \varphi_{n,m}^{(i)}\circ D=k^{p(i)}\partial_i\circ \varphi_{n,m}^{(i)}-\frac12\sum_{j\ne i}k^{1-p(j)}\frac{x_i+x_j}{x_i-x_j}(\varphi_{n,m}^{(i)}-\varphi^{(j)}_{n,m}).
 \end{equation}

After the gauge transformation with $$\Psi_0=\prod_{i<j}^{n+m}\left(\frac{x_ix_j}{(x_i-x_j)^2}\right)^{\frac{1}{2}k^{1-p(i)-p(j)}}$$ 
this leads us to the following quantum version of Moser's matrix  in the deformed trigonometric case
 \begin{equation}\label{qMt}
L_{ii}=k^{p(i)}\partial_i, \,\,\,\, L_{ij}=\frac12k^{1-p(j)}\frac{x_i+x_j}{x_i-x_j},\,\,i\ne j.
\end{equation}

Define also $(n+m)\times (n+m)$ matrix $M$ by
\begin{equation}
\label{moserMq}
M_{ij}=\frac{2k^{1-p(j)}x_ix_j}{(x_i-x_j)^2}, \,\, i\neq j, \quad M_{ii}=-\sum_{j\neq i}^{n+m}\frac{2k^{1-p(j)}x_ix_j}{(x_i-x_j)^2}.
\end{equation}

Let $e$ and $e^*$ be the same as in \ref{corol}, and $H$ is defined  by $H_{ii}=H_{n,m},\,\, H_{ij}=0,\,\,  i\neq j$, with $L_{n,m}$ defined by (\ref{anmi}). One can check that like in the rational case we have
the quantum Lax relation
$$[L,H]=[L,M],$$
leading to the following set of integrals. 

\begin{thm} The operators  
\begin{equation}\label{qMtin}
 L^{(r)}_{n,m}=e^* L^r e=\sum_{i,j=1}^{m+n}k^{-p(i)}(L^R)_{ij}
\end{equation}
are the commuting quantum integrals of the deformed CMS system (\ref{anmi}).
\end{thm}

To prove the commutativity we should consider the gauged version $\mathcal L$ of matrix $L$
 \begin{equation}\label{qMt}
\mathcal L_{ii}=k^{p(i)}\partial_i-\frac12\sum_{j\ne i}k^{1-p(j)}\frac{x_i+x_j}{x_i-x_j}, \,\,\,\, 
\mathcal L_{ij}=\frac12k^{1-p(j)}\frac{x_i+x_j}{x_i-x_j},\,\,i\ne j
\end{equation}
and define the operators
 \begin{equation}\label{qMtint}
 \mathcal L^{(r)}_{n,m}=e^*\mathcal L^r e
\end{equation}
with $\mathcal L^{(2)}_{n,m}$ being the quantum Hamiltonian of the deformed CMS system (\ref{defCMSt}).
Similarly to the rational case one can show that the  following diagram
$$
\begin{array}{ccc}
\bar\Lambda&\stackrel{\mathcal L^{(r)}}{\longrightarrow}&\bar\Lambda\\ \downarrow
\lefteqn{\varphi_{n,m}}& &\downarrow \lefteqn{\varphi_{n,m}}\\
\Lambda_{n,m}&\stackrel{\mathcal L^{(r)}_{n,m}}{\longrightarrow}& 
\Lambda_{n,m} \\
\end{array}
$$
is commutative. Since the operators $\mathcal L^{(r)}$ commute with each other, the same is true for $\mathcal L^{(r)}_{n,m},$
and hence for $L^{(r)}_{n,m}.$

\section{Rational $B$-type case}

For simplicity we will restrict ourselves with the rational case only. 

The rational CMS operator of type $B_N$ has the form
$$
H_N=\sum_{i=1}^N\frac{\partial^2}{\partial x^2_i}-\sum_{i<j}^N\frac{2k(k+1)}{(x_i-x_j)^2}-\sum_{i<j}^N\frac{2k(k+1)}{(x_i+x_j)^2}-\sum_{i=1}^N\frac{q(q+1)}{x_i^2}
$$
and depends on two parameters $k$ and $q.$
Its gauged version $\mathcal H_N=\delta H_N\delta^{-1}$ with 
$$\delta=\prod_{i<j}^N(x_i-x_j)^k(x_i+x_j)^k\prod_{i}^Nx_i^p$$  is
 $$
 \mathcal H_N=\sum_{i=1}^N\frac{\partial^2}{\partial x^2_i}-\sum_{i<j}^N\frac{2k}{x_i-x_j}\left(\frac{\partial}{\partial x_i}-\frac{\partial}{\partial x_j}\right)-\sum_{i<j}^N\frac{2k}{x_i+x_j}\left(\frac{\partial}{\partial x_i}+\frac{\partial}{\partial x_j}\right)
 $$
  \begin{equation}
  \label{CMbn}
  -\sum_{i=1}^N\frac{2q}{x_i}\frac{\partial}{\partial x_i}.
\end{equation}

The CMS operator (\ref{CMbn}) preserve the algebra of symmetric polynomials with respect to the group $W_N=S_N\ltimes \Bbb Z_2^N$
$$\Lambda_{N}=\Bbb C[x_{1},\dots, x_{N}]^{W_N},$$
generated (not freely) by $p_j(x) = x_1^{2j} + \dots + x_N^{2j}, \, j \in \mathbb Z_{\ge 0}.$ 
The group $W_N$ is generated by the reflections $$\sigma_{ij}^+:\, (x_i, x_j)\rightarrow (x_j, x_i), \,\,\,
\sigma_{ij}^-:\, (x_i, x_j)\rightarrow (-x_j, -x_i),\,\, 1 \leq i < j \leq N$$ and
$$\tau_i:\, x_i \rightarrow -x_i, \, i=1, \dots, N$$ (leaving the other coordinates untouched).

The Dunkl operators in this case are
 case have the form
\begin{equation}\label{Dun}
D_{i,N}=\frac{\partial}{\partial x_i}-k\sum_{j\ne i}\left(\frac{1}{x_i-x_j}{(1-\sigma_{ij}^{+})}+\frac{1}{x_i+x_j}{(1-\sigma^{-}_{ij})}\right)-\frac{p}{x_{i}}(1-\tau_i),
\end{equation}
 where $ i=1,2,\dots,N.$ 
These operators commute \cite{Dun} and can generate the integrals of the CMS operator  by
\begin{equation}
  \label{heckcmt}
 \mathcal L^{(2r)}_{N}=Res \,(D_{1,N}^{2r}+\dots+D^{2r}_{N,N}),
\end{equation}
where $Res$ means the operation of restriction on the space of symmetric polynomials \cite{H}, with $\mathcal L^{(2)}_{N}=\mathcal H_N$ given by (\ref{CMbn}).

For every natural $N$ there is a homomorphism $\varphi_N:  \bar\Lambda \rightarrow \Lambda_{N}:$
\begin{equation}\label{phin}
\varphi_N (p_j)= x_1^{2j }+ \dots + x_N^{2j}, \, j \in \mathbb Z_{\ge 0}.
\end{equation}

The infinite dimensional Dunkl operators of $B$-type are the operators $D: \bar\Lambda[x]\rightarrow \bar\Lambda[x]$  defined by 
\begin{equation}\label{Duninf}
D=\partial-2k\Delta-\frac{q}{x}(1-\tau).
\end{equation}
Here the differentiation $\partial$ in  $\bar\Lambda[x]$ is defined by the formulae
$$
\partial(x)=1, \,\, \partial (p_l)=2lx^{2l-1},  l\in \Bbb Z_{\ge 0},
$$
the operator $\Delta:\,\bar\Lambda[x] \rightarrow \bar\Lambda[x]$ is defined by
$$
\Delta(x^lf)=\Delta(x^l)f,\,  \Delta(1)=0, \,f\in \bar\Lambda ,\,l\in \Bbb Z_{\ge 0}
$$
with
$$
\Delta(x^{2l})=x^{2l-1}p_0+x^{2l-3}p_1+\dots+ x^3p_{l-2}+xp_{l-1}-lx^{2l-1},
$$ 
$$
\Delta(x^{2l-1})=x^{2l-2}p_0+x^{2l-4}p_1+\dots+ x^2p_{l-2}+p_{l-1}-lx^{2l-2},\,  l>0,
$$
and the involution $\tau$ is defined by
$$
\tau(x^l f)=(-x)^lf,\,\,f\in\bar\Lambda.
$$
 Let $\varphi_{i,N} : \bar\Lambda[x]\longrightarrow \Lambda_{N}[x_i]$ be the homomorphism such that
$$
\varphi_{i,N}(x)=x_i,\,\varphi_{i,N}(p_l)=x_1^{2l}+\dots+x_N^{2l}, \, l \in \mathbb Z_{\ge 0}.
$$ 
One can show that the following diagram
\begin{equation}\label{commutdun}
\begin{array}{ccc}
\bar\Lambda[x]&\stackrel{D}{\longrightarrow}&\bar\Lambda[x]\\ \downarrow
\lefteqn{\varphi_{i,N}}& &\downarrow \lefteqn{\varphi_{i,N}}\\
\Lambda_{N}[x_i]&\stackrel{D_{i,N}}{\longrightarrow}& 
\Lambda_{N}[x_i],\\
\end{array}
\end{equation}
 where $D_{i,N}$ are the Dunkl operators (\ref{Dun}), is commutative.

 Define a linear operator $E : \bar\Lambda[x]\longrightarrow \bar\Lambda$  by the formulae 
$$E(x^{2l}f)=p_{l}f, \,\,E(x^{2l+1}f)=0,\, f\in\bar\Lambda, \,\,l\in \Bbb Z_{\ge 0},$$
and the operators $\mathcal {L}^{(r)}: \bar\Lambda\longrightarrow \bar\Lambda,\,\, r \in \mathbb Z_+$  by
 \begin{equation}\label{Lr}
 \mathcal{L}^{(r)}=E\circ D^{2r},
 \end{equation}
 where the action of the right hand side is restricted to $\bar\Lambda.$

The claim is that these operators give the quantum CMS integrals at infinity in the rational $BC$ case.
\begin{thm}\label{Heck} 
The differential operators $\mathcal{L}^{(r)}$ commute with each other:
$$[\mathcal{L}^{(r)}, \mathcal{L}^{(s)}]=0.$$
The operator $\mathcal L^{(2)}$ has the following explicit form:
$$
\mathcal L^{(2)}=8\sum_{a,b\geq 1}p_{a+b-1}\partial_{a}\partial_{b}-4k\sum_{a,b\geq 0}p_{a}p_b \partial_{a+b+1} +4k\sum_{a\ge0}(a+1)p_a\partial _{a+1}
$$
 \begin{equation}\label{inf2ratio}
+2\sum_{a\ge0}(2a+1)p_a\partial _{a+1}-4q\sum_{a\ge0}p_a\partial _{a+1}
\end{equation}
with  $\partial_a = a\partial/\partial p_a,$ and coincides with the rational CMS operator of $B$-type at infinity.
\end{thm}

The explicit form (\ref{inf2ratio}) is in agreement with formula (32) from \cite{SV3} and follows from the relations
$$
(E\circ \Delta\circ\partial) (p_a)=2a(p_{a-1}p_0+\dots+p_0p_{a-1}-ap_{a-1}),
$$
$$
(E\circ\frac{1}{x}(1-\tau)\circ\partial) (p_a)=4ap_{a-1},\,\,\,\,E\circ\partial^2(p_a)=2a(2a-1)p_{a-1},
$$
$$
(E\circ\partial^2) (p_ap_b)=2a(2a-1)p_{a-1}p_b+2b(2b-1)p_{b-1}p_a+8abp_{a+b-1}.
$$

Now let's apply this to the deformed case.
The deformed rational  CMS operator of type $B_{n,m}$ has the form \cite{SV}
$$
 H_{n,m}= -\left(\frac{\partial^2}{{\partial
x_{1}}^2}+\dots +\frac{\partial^2}{{\partial x_{n}}^2}\right)
-k\left(\frac{\partial^2}{{\partial y_{1}}^2}+ \dots
+\frac{\partial^2}{{\partial y_{m}}^2} \right)
$$
$$ 
+\sum_{i<j}^{n}\left(\frac{2k(k+1)}{(x_{i}-x_{j})^2}+\frac{2k(k+1)}{(x_{i}+x_{j})^2}\right)
+\sum_{i<j}^{m}\left(\frac{2(k^{-1}+1)}{(y_{i}-y_{j})^2}+\frac{2(k^{-1}+1)}{(y_{i}+y_{j})^2}\right)
$$
\begin{equation}
\label{bcnmR} 
 +\sum_{i=1}^{n}\sum_{j=1}^{m}\left(\frac{2(k+1)}{(x_{i}-y_{j})^2}+
\frac{2(k+1)}{(x_{i}+y_{j})^2}\right) +\sum_{i=1}^n
\frac{q(q+1)}{x_{i}^2}+\sum_{j=1}^m \frac{ks(s+1)}{y_{j}^2},
\end{equation}
 where the parameters $k,q,s$ satisfy the relation
\begin{equation}
\label{rel}
 2q+1=k(2s+1).
\end{equation}

Let $x_{n+i}=y_i, \, i=1,\dots,m$ as before and introduce the multiplicity function 
$m(i)=q$ for $i=1,\dots,n$ and $m(i)=s$ for $i=n+1,\dots, n+m.$

Define  $\varphi^{(i)}_{n,m}:\bar\Lambda[x]\longrightarrow \Bbb C[x_1,\dots,x_{n+m}]$ by 
$\varphi^{(i)}_{n,m}(x)=x_i$ and 
\begin{equation}
\label{varpbc}
\varphi^{(i)}_{n,m}(p_l)=\sum_{i=1}^{n+m} k^{-p(i)}x_i^{2l}=\sum_{i=1}^{n} x_i^{2l}+ \frac{1}{k}\sum_{i=n+1}^{n+m} x_i^{2l},
\end{equation}
for all $l\in \mathbb Z_{\ge 0}.$ Let also as before $\tau_i$ be the homomorphism of $C[x_1,\dots,x_{n+m}]$ changing the sign of $x_i.$

 \begin{proposition} We have the following relation on  $\bar\Lambda[x]$
$$
 \varphi_{n,m}^{(i)}\circ D=k^{p(i)}\frac{\partial}{\partial x_i}\circ \varphi_{n,m}^{(i)}-\frac{k^{p(i)}m(i)}{x_i}(1-\tau_i)\varphi_{n,m}^{(i)}
 $$
  \begin{equation}
\label{commnet}
-\sum_{j\ne i}\frac{k^{1-p(j)}}{x_i-x_j}(\varphi_{n,m}^{(i)}-\varphi^{(j)}_{n,m}) -\sum_{j\ne i}\frac{k^{1-p(j)}}{x_i+x_j}(\varphi_{n,m}^{(i)}-\tau_j \varphi^{(j)}_{n,m}).
 \end{equation}
 \end{proposition}

 The corresponding quantum version of Moser's matrix has the block form
 $$
L= \begin{pmatrix}
A  &  B   \\
-B  &  -A  
\end{pmatrix}
$$ 
with the following $(n+m)\times (n+m)$ matrices $A$ and $B$:
 \begin{equation}
\label{moserLqA}
A_{ii}=k^{p(i)}\frac{\partial}{\partial x_i},\quad A_{ij}=\frac{k^{1-p(j)}}{x_i-x_j}, \,\, i\neq j,
\end{equation}
\begin{equation}
\label{moserLqB}
B_{ii}=\frac{k^{p(i)}m(i)}{x_i},\quad B_{ij}=\frac{k^{1-p(j)}}{x_i+x_j}, \,\, i\neq j,
\end{equation}
which are the deformed versions of matrices from \cite{Y}.

Let $e=(1,\dots,1)^T$ and $e^*_i=e^*_{n+m+i}=k^{-p(i)}$ for $i=1,\dots, (n+m)$.

\begin{thm} The operators  
\begin{equation}\label{qMtin}
 L^{(l)}_{n,m}=e^* L^{2l} e,
\end{equation}
are the commuting quantum integrals of the deformed CMS system in the rational $B_{n,m}$ case (\ref{bcnmR}).
\end{thm}
 
\section{Trigonometric $BC$ case}

The trigonometric $BC_N$ CMS operator depends on 3 parameters $k,p,q$ and in the exponential coordinates has the form
$$
H_{N}=\sum_{i=1}^N (x_i\frac{\partial}{\partial x_{i}})^2
  -\sum_{i<j}^{n}\left(\frac{2k(k+1)x_ix_j}{(x_{i}-x_{j})^2}+\frac{2k(k+1)x_ix_j}{(x_{i}x_{j}-1)^2}\right)
$$
\begin{equation}
\label{bNi} 
-\sum_{i=1}^n \left(\frac{p(p+2q+1)x_i}{(x_i-1)^2}+\frac{4q(q+1)}{(x_i^2-1)^2}\right),
\end{equation}
or, after a gauge transformation and using $\partial_i=x_i\frac{\partial}{\partial x_{i}}$,
$$
{\mathcal H}_N=\sum_{i=1}^N \partial_{i}^2-k\sum_{1\le i < j\le N} \frac{x_{i}+x_{j}}{x_{i}-x_{j}}(\partial_{i}-\partial_{j})-k\sum_{1\le i < j\le N} \frac{x_{i}x_{j}+1}{x_{i}x_{j}-1}(\partial_{i}+\partial_{j})
$$
\begin{equation}
\label{CMSt}
 -\sum_{i=1}^N\left(p\frac{x_{i}+1}{x_{i}-1}+2q\frac{x^2_{i}+1}{x^2_{i}-1}\right)\partial_{i}.
\end{equation}
The operator ${\mathcal H}_N$ preserves the algebra $\Lambda^W_N$ of $W_N$-invariant functions, where the action of Weyl group $W_N=S_N\ltimes \Bbb Z_2^N.$ This group is generated by $s_{ij}^{\pm}$ and $t_i, \, i=1,\dots,N,$ 
acting according  to $$s_{ij}^{\pm}(x_i,x_j)=(x_j^{\pm 1}, x_i^{\pm 1}), \quad t_i(x_i)=x_i^{-1}, \, \,\,i=1,\dots, N$$ 
(other coordinates are unchanged).
The algebra $\Lambda^W_N$ is generated by the invariants $$p_l=x_1^l+x_1^{-l}+\dots+x_N^l+x_N^{-l},\,\, l \in \mathbb Z_{>0}.$$

The corresponding Dunkl-Heckman operators $D_{i,N}: \Lambda^W_N[x_i, x_i^{-1}] \rightarrow\Lambda^W_N[x_i, x_i^{-1}]$ have the form
$$
D_{i,N}=\partial_i-\frac{1}{2}k\sum_{j\ne i}^N\left(\frac{x_i+x_j}{x_i-x_j}{(1-s_{ij}^{+})}+\frac{x_ix_j+1}{x_ix_j-1}{(1-s^{-}_{ij})}\right)
$$
\begin{equation}\label{DunN}
-\frac{1}{2}p\frac{x_i+1}{x_{i}-1}(1-t_i)-q\frac{x_i^2+1}{x_i^2-1}(1-t_i).
\end{equation}  

The CMS integrals can be given by Heckman's formula
$$ \mathcal L^{(2r)}_{N}=Res \,(D_{1,N}^{2r}+\dots+D^{2r}_{N,N}),
$$
where $Res$ means the operation of restriction on the space of $W_N$-invariant functions \cite{H}. One can check that $\mathcal L^{(2)}_{N}={\mathcal H}_N$ is the CMS operator (\ref{CMSt}).

Consider the algebra $\bar\Lambda$ freely generated by  $p_i, \, i \in \mathbb Z_{\geq 0}$ as before.
Define the infinite-dimensional version of the $BC$ Dunkl-Heckman operator $D: \bar\Lambda[x,x^{-1}] \rightarrow \bar\Lambda[x,x^{-1}]$ as
\begin{equation}\label{DunNN}
D=\partial -\frac{1}{2}k \Delta -\frac{1}{2} p\frac{x+1}{x-1}(1-t) -q \frac{x^2+1}{x^2-1}(1-t),
\end{equation}
where the differentiation $\partial$ is defined by $\partial (x)=x, \partial p_l=l (x^l-x^{-l}), \, l \in \mathbb Z_{\geq 0},$
and the homomorphisms of $\bar\Lambda$-modules $\Delta$ and $t$ defined by $\Delta\, (1)=0,$
$$\Delta (x^l)=(p_0-2l-1)x^l-2\sum_{j=1}^{l-1} x^{l-2j}-x^{-l}+2\sum_{j=1}^{l-1}p_j x^{l-j} + p_l,$$
$$\Delta (x^{-l})=-(p_0-2l-1)x^{-l}+2\sum_{j=1}^{l-1} x^{l-2j}+x^{l}-2\sum_{j=1}^{l-1}p_j x^{-l+j} - p_{l},\,\, l>0,$$
$$
t(x)=x^{-1}, \quad t(p_l)=p_{l}.
$$

One can check that the diagram
$$
\begin{array}{ccc}
\bar\Lambda[x,x^{-1}]&\stackrel{D}{\longrightarrow}&\bar\Lambda[x,x^{-1}]\\ \downarrow
\lefteqn{\varphi_{i,N}}& &\downarrow \lefteqn{\varphi_{i,N}}\\
\Lambda^W_{N}[x_i,x_i^{-1}]&\stackrel{D_{i,N}}{\longrightarrow}& 
\Lambda^W_{N}[x_i,x_i^{-1}],\\
\end{array}
$$
 is commutative, where 
 $\varphi_{i,N}(x)=x_i$ and 
 $$\varphi_{i,N}(p_l)=x_1^l+x_1^{-l}+\dots+x_N^l+x_N^{-l},\, l \ge 0$$
 (in particular, $p_0$ is specialised to $2N$).

Define the homomorphism of $\bar\Lambda$-modules $E: \bar\Lambda[x,x^{-1}] \rightarrow \bar\Lambda$ by
$$E(x^j)=p_{|j|},\,\, j \in \mathbb Z.$$
The CMS integrals of $BC$-type at infinity can be defined now by the formula
 \begin{equation}\label{Lrbc}
 \mathcal{L}^{(2r)}=E\circ D^{2r},
 \end{equation}
 where the action of the right hand side is to be restricted to $\bar\Lambda.$
 For $r=1$ we have the $BC$ operator at infinity
 $$
\mathcal L^{(2)}= 4\sum_{a,b\ge 1} (p_{a+b}-p_{a-b})\partial_a\partial_b 
+2\sum_{a\ge 1}(ak+a+k+h)p_a\partial_a 
$$
\begin{equation}\label{Bcinfty}
+ 2(k-q) \sum_{a\ge 2}(\sum_{j=1}^{a-1}p_{a-2j})\partial_a - p \sum_{a\ge 2}(\sum_{j=1}^{2a-1}p_{a-j})\partial_a
-2k\sum_{a\ge 2}(\sum_{j=1}^{a-1}p_{j}p_{a-j})\partial_a,
\end{equation} 
 where as usual $\partial_a= a \frac{\partial}{\partial p_a}$ and we used the notation from \cite{SV3} $$h=-kp_0-\frac{1}{2}p-q$$ and defined $p_k:=p_{|k|}$ for all $k\in \mathbb Z.$
Note that the comparison with the formulae in \cite{SV3} is not easy since 
the variables there correspond to the different choice of invariants in $\Lambda^W_N:$
$$p_l=\sum u_i^l, \,\, u_i=\frac{1}{2}(x_i+x_i^{-1}-2).$$

Consider now briefly the deformed case.
The corresponding CMS operator \cite{SV}
in the exponential coordinates has the form 
$$
H_{n,m}=\sum_{i=1}^n (x_i\frac{\partial}{\partial x_{i}})^2
+k\sum_{j=1}^m (y_j\frac{\partial}{\partial y_{j}})^2  
-\sum_{i<j}^{n}\left(\frac{8k(k+1)x_ix_j}{(x_{i}-x_{j})^2}+\frac{8k(k+1)x_ix_j}{(x_{i}x_{j}-1)^2}\right)
$$
\begin{equation}
\label{bnmi} 
-\sum_{i<j}^{m}\left(\frac{8(k^{-1}+1)y_iy_j}{(y_{i}-y_{j})^2}+\frac{8(k^{-1}+1)y_iy_j}{(y_{i}y_j-1)^2}\right)
-\sum_{i=1}^{n}\sum_{j=1}^{m}\frac{8(k+1)x_i y_j}{(x_{i}-y_{j})^2}
\end{equation}
$$
-\sum_{i=1}^n \left(\frac{4p(p+2q+1)x_i}{(x_i-1)^2}+\frac{16q(q+1)x_i^2}{(x_i^2-1)^2}\right)
-\sum_{j=1}^m \left(\frac{4kr(r+2s+1)y_j}{(y_j-1)^2}+\frac{16ks(s+1)y_j^2}{(y_j^2-1)^2}\right),
$$
 where the parameters $k,p,q,r,s$ satisfy the relations
 \begin{equation}
\label{rel}
p=kr,\quad 2q+1=k(2s+1).
\end{equation}


Denote as before $x_{n+i}=y_i, \, i=1,\dots,m,$ $\partial_j = x_j\frac{\partial}{\partial x_j}$ and introduce the multiplicity functions 
$\mu(i)=p,\,\nu(i)=q$ for $i=1,\dots,n$ and  $\mu(i)=r,\,\nu(i)=s$ for $i=n+1,\dots, n+m.$ 

 The quantum Moser's matrix in this case also has the form
 $$
L= \begin{pmatrix}
A  &  B   \\
-B  &  -A  
\end{pmatrix}
$$ 
with the following $(n+m)\times (n+m)$ matrices $A$ and $B$:
 \begin{equation}
\label{moserLqA}
A_{ii}=k^{p(i)}\partial_i,\quad A_{ij}=\frac{k^{1-p(j)}(x_i+x_j)}{2(x_i-x_j)}, \,\, i\neq j,
\end{equation}
\begin{equation}
\label{moserLqB}
B_{ii}=\frac{k^{p(i)}\mu(i)(x_i+1)}{2(x_i-1)}+\frac{k^{p(i)}\nu(i)(x^2_i+1)}{x^2_i-1},\quad B_{ij}=\frac{k^{1-p(j)}(x_ix_j+1)}{2(x_ix_j-1)}, \,\, i\neq j.
\end{equation}

The commuting quantum integrals of the deformed CMS system (\ref{bnmi}) now can be constructed as 
\begin{equation}\label{qMtin}
 L^{(2l)}_{n,m}=e^* L^{2l} e,
\end{equation}
where as before $e=(1,\dots,1)^T$ and $e^*_i=e^*_{n+m+i}=k^{-p(i)}$ for $i=1,\dots, (n+m)$.

\section{Concluding remarks}

We have shown how Dunkl operator at infinity leads to the quantum Moser matrix and Lax pair for the deformed 
CMS systems related to classical series of Lie superalgebras. A simple form of the corresponding quantum Moser matrix suggests 
that it might be possible to guess it for the deformed CMS systems related to the exceptional Lie superalgebras \cite{SV}.\footnote{Oleg Chalykh has informed us that, at least in the rational case, there is a relatively simple way to prove the integrability of all deformed CMS systems, including exceptional ones, using the theory of rational Cherednik algebras \cite{BC}.}

Another open question is about elliptic version. The elliptic Dunkl operators were studied in \cite{BFV} and were used to construct 
the integrals of the elliptic CMS systems in \cite{EFMV}. The construction is not as straightforward as in trigonometric case and involves the integrals of the corresponding classical system. The question is if the methods of our paper could be modified to this case.

Finally it is interesting to understand the precise relation of our formulae for quantum CMS integrals at infinity in trigonometric type $A$ case with the results of the recent paper \cite{NS} by Nazarov and Sklyanin, whose main tool was the quantum Lax operator for the periodic Benjamin-Ono equation, which they have introduced.\footnote{After the first version of the present paper had appeared in the ArXiv, Evgeni Sklyanin informed us that he also came with Maxim Nazarov to the idea of using Dunkl operator technique for construction of the CMS integrals at infinity (in the type $A$ case).} We believe that their integrals (which do not depend on $p_0$) are simply related to the stable integrals $\mathcal{H}_{k}^{(r)}$ from our recent paper \cite{SV6}, which were constructed using the infinite-dimensional version of Polychronakos operator (rather than Dunkl-Heckman operator used in the present paper). Note that the relation between $\mathcal{H}_{k}^{(r)}$  and our quantum CMS integrals (\ref{LrH}) is non-trivial (see the formulae in section 5 of \cite{SV6}).

\section{Acknowledgements}

We are grateful to O.A. Chalykh and M.V. Feigin for useful discussions.

This work was partly supported by the EPSRC (grant EP/J00488X/1). ANS is grateful to Loughborough University for the hospitality during the autumn semesters 2012 and 2013.

\end{document}